\documentclass[copyright,creativecommons]{eptcs}
\providecommand{\event}{PLACES 2020}

\usepackage{amsmath,amssymb}
\usepackage{graphicx}
\usepackage{QED}
\usepackage{color}
\usepackage{qtree}

\newtheorem{theorem}{Theorem}
\newtheorem{definition}[theorem]{Definition}
\newtheorem{lemma}[theorem]{Lemma}
\newtheorem{proposition}[theorem]{Proposition}
\newtheorem{corollary}[theorem]{Corollary}
\newtheorem{example}[theorem]{Example}

\newcommand{\keyword}[1]{\mathsf{#1}}

\newcommand{\End}{\keyword{end}}
\newcommand{\inp}[2]{?{#1}.#2}
\newcommand{\outp}[2]{!{#1}.#2}
\newcommand{\Int}{\keyword{int}}
\newcommand{\rec}[2]{\mu{#1}.{#2}}

\newcommand{\SType}{\operator{SType}}
\newcommand{\Type}{\operator{Type}}

\newcommand{\D}{\mathcal{D}}
\newcommand{\E}{\mathcal{E}}
\newcommand\ST{\mathcal{S}}

\newcommand{\operator}[1]{\mathsf{#1}}
\newcommand{\fv}{\operator{fv}}
\newcommand{\dom}{\operator{dom}}
\newcommand{\dual}[1]{\overline{#1}} 
\newcommand{\dualrel}{\asymp} 

\newcommand{\cdual}[1]{\mathsf{dual}_{\mathsf{LMN}}(#1)}
\newcommand{\coidual}[1]{\operator{dual}(#1)}
\newcommand\embed[1]{\lfloor#1\rfloor}
\newcommand{\lmdual}{\mathsf{dual}_{\mathsf{LM}}}
\newcommand{\lmdualp}{\mathsf{dual}_{\mathsf{LMP}}}

\newcommand{\dualofop}{\operator{\mathsf{dual}_{BH}}}
\newcommand{\mcduals}[2][\sigma]{\dualofop({#2},{#1})}
\newcommand{\mcdual}[1]{\dualofop({#1})}
\newcommand{\subst}[2]{[#1/#2]} 
\newcommand{\lmsubst}[2]{\{#1/#2\}} 

\newcommand{\sequiv}{\approx}
\newcommand{\Fequiv}[1]{F_\sequiv(#1)}
\newcommand{\Fdual}[1]{F_\perp(#1)}

\newcommand{\mclo}[2]{\operator{mclo}(#1,#2)}
\newcommand{\mcl}[1]{\operator{mclo}(#1)}
\newcommand{\treeof}[1]{\operator{treeof}(#1)}
\newcommand{\varseq}{\mathcal X}

\newcommand\Sgen{{S}_{\operator{gen}}}

\newcommand{\istailrec}[2][\varseq]{{#1}\vdash{#2}}

\newcommand{\Empty}{\varepsilon}

\newcommand{\grmeq}{\; ::= \;}
\newcommand{\grmor}{\; \mid \;}

\title{Duality of Session Types: The Final Cut}
\author{Simon J.\ Gay
\institute{School of Computing Science}
\institute{University of Glasgow, UK}
\email{Simon.Gay@}
\email{glasgow.ac.uk}
\and
Peter Thiemann
\institute{Institut f\"{u}r Informatik}
\institute{University of Freiburg, Germany}
\email{thiemann@}
\email{informatik.uni-freiburg.de}
\and
Vasco T.\ Vasconcelos
\institute{Faculdade de Ci\^{e}ncias}
\institute{University of Lisbon, Portugal}
\email{vmvasconcelos@}
\email{ciencias.ulisboa.pt}
}
\def\titlerunning{Duality of Session Types: The Final Cut}
\def\authorrunning{S.\ J.\ Gay, P.\ Thiemann \& V.\ T.\ Vasconcelos}
\begin{document}
\maketitle

\begin{abstract}
  Duality is a central concept in the theory of session types.  Since
  a flaw was found in the original definition of duality for recursive types,
  several other definitions have been published. As their connection
  is not obvious, we compare the competing definitions, discuss
  tradeoffs, and prove some equivalences. Some of the results are
  mechanized in Agda.
\end{abstract}


\section{Introduction}
\label{sec:intro}
Duality is a central concept in the theory of session types. If $S$ is a session type describing a two-party interaction from the viewpoint of one party, then $\dual{S}$ describes the interaction from the viewpoint of the other party. For example, $S = \rec{X}{\inp{\Int}{X}}$ describes indefinitely receiving integers, and its dual $\dual{S} = \rec{X}{\outp{\Int}{X}}$ describes indefinitely sending integers. If the users of the two endpoints of a channel follow types $S$ and $\dual{S}$, respectively, then correct communication takes place.

The original papers on session types
\cite{DBLP:conf/concur/Honda93,HondaVK98,TakeuchiHK94} define the dual
$\dual{S}$ of a session type $S$ by structural induction on $S$:
\begin{align*}
  \dual{\End} = \End
  &&
  \dual{\outp{T}{S}} = \inp{T}{\dual{S}}
  &&
  \dual{\inp{T}{S}} =\outp{T}{\dual{S}}
\end{align*}

Recursion is only introduced in the last paper in the series,
\cite{HondaVK98}, where recursive session types are handled via the
following rules.
\begin{align*}
  \dual{X} = X && \dual{\rec{X}{S}} = \rec{X}{\dual{S}}
\end{align*}

With this definition, indeed we have
$\dual{\rec{X}{\inp{\Int}{X}}} = \rec{X}{\outp{\Int}{X}}$, given that
duality exchanges input and output. Gay \& Hole \cite{GayH99,GayH05}
define a more general duality \emph{relation} $\perp$ so that, for
example,
$\rec{X}{\inp{\Int}{X}} \perp
\rec{X}{\outp{\Int}{\outp{\Int}{X}}}$. The definition is
coinductive. This relation gives greater flexibility in typing
derivations and follows the idea that duality is a behavioural
relation on automata. The relationship between the duality function
and the duality relation is intended to be that $\dual{S} \perp S$ for
every session type $S$.

Bernardi \& Hennessy \cite{BernardiH14,BernardiH16} show that the
duality function $\dual{(\,\cdot\,)}$ violates the duality relation
for recursive session types when the recursion variable can occur as the type of a message, as in $S = \rec{X}{\inp{X}{X}}$. In this example, we have $\dual{S} = \rec{X}{\outp{X}{X}}$. Noting that an occurrence of $X$ stands for the whole $\mu$ type, the type of the message in $S$ is $S$ but the type of the message in $\dual{S}$ is $\dual{S}$. In other words, we have dual types in which the type of the message being sent is not the same as the type of the message being received, which violates soundness of any type system that uses this definition of duality. We refer to $\dual{(\,\cdot\,)}$ as \emph{naive duality} because it initially seems reasonable but is not correct in all situations.

One way to solve this problem is to require that recursion variables
only occur in tail position in a session type, such as $X$ in
$\inp{T}{X}$. As far as we know, almost all papers that use naive
duality can be fixed by restricting to tail recursion, because their
examples and applications are all tail recursive. One exception is a
paper by Vasconcelos~\cite{Vasconcelos12}, which has an interesting
application of the type $\rec{X}{\outp{X}{X}}$ to encode replication,
but that could be easily solved with a tail recursive type at the
expense of creating an extra channel. Bernardi \& Hennessy give
examples of pi-calculus processes that can only be typed by using
non-tail-recursive session types, but they are specially crafted for
the purpose. Nevertheless, it is more satisfactory to have a duality
function that works for all session types.

Bernardi \& Hennessy \cite{BernardiH16} give an alternative, correct duality function and justify it with respect to their model of session types which is based on a theory of contracts. The key idea is that a session type can be converted into an equivalent type in which all message types are closed. In Section~\ref{sec:duality-bh} we present their definition and a variation of it, and reformulate their correctness result in a standard model of recursive types.


Bernardi, Dardha, Gay \& Kouzapas \cite{BernardiDGK14} discuss several definitions of duality, focusing on the fact that there can be different sound definitions which give rise to different typing relations. One of their definitions is that of Bernardi \& Hennessy \cite{BernardiH16}. They point out that some results claimed by Gay \& Hole \cite{GayH05} are false for non-tail-recursive types.  

Lindley \& Morris \cite{LindleyM16} give another definition of the duality function. It maps a type variable $X$ in tail position to a \emph{negative type variable} $\dual{X}$, but in a message position it remains as $X$. As well as being a technical convenience, negative variables allow interesting types such as $\rec{X}{\outp{\dual{X}}{X}}$. Lindley \& Morris justify their definition on general type-theoretic grounds, but do not directly prove its correctness with respect to the duality relation. In Section~\ref{sec:duality-lm} we do so, as well as giving an equivalent and arguably simpler variation, and another variation that turns out to be equivalent to the Bernardi-Hennessy definition.


As well as proving the results mentioned above on paper, we have begun work on mechanising them in Agda. We summarise the mechanisation in Section~\ref{sec:mechanized-results}.

\section{Basic Definitions about Session Types}
\label{sec:definitions}

We work formally with a subset of session types, consisting of input
and output (no branch or select), and $\Int$ as a representative data
type. All definitions and proofs can be straightforwardly extended to
cover branch and select; reasoning about equivalence and duality of
session types is not affected by the details of data types.

\begin{definition}[Types and Session Types]
  \label{def:types}
  Let $X,Y,Z$ range over a denumerable set of \emph{type variables}.
  Types ($T, U$) and session types ($R, S$) are defined by
  \begin{align*}
    T, U & \grmeq  \Int \grmor S &
                                   R, S & \grmeq \End \grmor \inp{T}{S} \grmor \outp{T}{S} \grmor X \grmor \rec{X}{S}
  \end{align*}
  Session types must be \emph{contractive}, meaning that they must not
  contain sub-expressions of the form
  $\rec{X}{\rec{X_1}{\ldots\rec{X_n}{X}}}$ for $n\ge0$.
The expression $\rec XS$ binds type variable~$X$ with scope $S$.
The set $\fv(T)$ of \emph{free type variables} in a type~$T$ is
defined as usual, and so is \emph{$\alpha$-congruence}. 
%
The set of closed session types is denoted by $\SType$ and the set of
closed types is denoted by $\Type$, so that
$\Type = \SType \cup \{ \Int \}$.  We identify types that are
$\alpha$-congruent and follow the Barendregt convention on
variables~\cite{DBLP:books/daglib/0067558}.
\end{definition}

In the sequel, we use term \emph{type} for any contractive type
generated by the grammar for $T$. When we mean a closed type, we shall
speak of $T\in\Type$. The same reasoning applies to session types,
where the term \emph{session type} denotes a contractive type
generated by the grammar for $S$, and $S\in\SType$ denotes a closed
session type.

\begin{definition}[Substitution]
  \label{def:substitution}
  The result of substituting type~$U$ for the free occurrences of
  variable~$X$ in type~$T$---notation $T\subst{U}{X}$---is defined
  inductively as follows, where $Y\ne X$.
  \begin{align*}
    \Int\subst{U}{X} & =\, \Int &
    (\inp{S}{T})\subst{U}{X} & =\, \inp{S\subst{U}{X}}{T\subst{U}{X}} &
    X\subst{U}{X} & =\, U&
    (\rec{X}{S})\subst{U}{X} & =\, \rec{X}{S}
    \\
    \End\subst{U}{X} &=\, \End &
    (\outp{S}{T})\subst{U}{X} & =\, \outp{S\subst{U}{X}}{T\subst{U}{X}} &
    Y\subst{U}{X} & =\, Y &
    (\rec{Y}{S})\subst{U}{X} & =\, \rec{Y}{S\subst{U}{X}}
  \end{align*}
\end{definition}

Closed session types are interpreted as regular trees in the standard
way presented by Pierce \cite[Chapter 21]{Pierce}.  A regular tree is
a (possibly infinite) tree with a finite number of distinct subtrees.

\begin{definition}[Types as Trees]
  \label{def:treeof}
  Types are represented by regular trees whose nodes are taken from
  the set $\{\Int,\End,!,?\}$, $\Int$ and $\End$ have no descendants,
  $!$ and $?$ have two descendants, and $\Int$ can only occur as root
  or at the immediate left of $!$ or $?$.
  We write $\treeof T$ for the tree representation of~$T$.
\end{definition}


\begin{example}
  Let $S$ be the session type $\rec X{\outp XX}$. The regular tree
  $t=\treeof S$ can be depicted as below left. The tree $u$ such that $t
  \dualrel u$ can be depicted as below right. 

  \hfill
  \Tree[.$!$ [.$!$ [.$!$ t t ] [.$!$ t t ] ] [.$!$ [.$!$ t t ] [.$!$ t t ] ]]
  \hfill
  \Tree[.$?$ [.$!$ [.$!$ t t ] [.$!$ t t ] ] [.$?$ [.$!$ t t ] [.$?$ t u ] ]]
  \hfill
  \hspace{0ex}
\end{example}

Equivalence of session types is equality of trees. We give a
coinductive syntactic characterisation of equivalence.

\begin{definition}[Syntactic Equivalence of Types]
  If $\E$ is a relation on $\Type$ then $\Fequiv{\E}$ is the relation on $\Type$ defined by:
  \begin{align*}
    \Fequiv{\E} & = \{ (\End,\End) \} \\
    & \cup\; \{ (\Int,\Int) \} \\
    & \cup\; \{ (\inp{T_1}{S_1},\inp{T_2}{S_2}) ~|~ (T_1,T_2), (S_1,S_2) \in\E \} \\
    & \cup\; \{ (\outp{T_1}{S_1},\outp{T_2}{S_2}) ~|~ (T_1,T_2), (S_1,S_2) \in\E \} \\
    & \cup\; \{ (S_1,\rec{X}{S_2}) ~|~ (S_1,S_2\subst{\rec{X}{S_2}}{X}) \in\E \} \\
    & \cup\; \{ (\rec{X}{S_1},S_2) ~|~ (S_1\subst{\rec{X}{S_1}}{X},S_2) \in\E ~\text{and}~ S_2\not=\rec{Y}{S_3} \}
  \end{align*}
  A relation $\E$ on $\Type$ is a \emph{type bisimulation} if
  $\E \subseteq \Fequiv{\E}$. Syntactic equivalence of types, $\sequiv$,
  is the largest type bisimulation.
\end{definition}

\begin{proposition}[Type equivalence is tree equality~\cite{Pierce}]
  \label{prop:type-equi-tree-equality}
  Let $T, U \in\Type$.  Then $T\sequiv U$ if and only if $\treeof{T} =
  \treeof{U}$.
\end{proposition}


The duality relation is defined on regular trees.

\begin{definition}[Duality on Trees]
  \label{def:tree-duality}
  Two trees, $s$ and $t$, are related by duality---notation
  $s\dualrel t$---if they have the same structure and, for each pair of
  corresponding nodes, if the nodes are \emph{in the right spine of
    the tree} they are related as below, otherwise they are the same.
  \begin{align*}
    \Int\leftrightarrow\Int && \End\leftrightarrow\End && ? \leftrightarrow{} ! && ! \leftrightarrow{} ?
  \end{align*}
Because $\dualrel$ is bijective and every tree is related to some other (unique) tree, we can also regard it as a self-inverse function, which we denote by $\coidual{\cdot}$.
\end{definition}


Proceeding as for type equivalence, we now give a coinductive
syntactic characterisation of the duality relation, restricting attention to session types because $\Int$ can only occur in message positions, where duality is never applied. This principle is applied to all of our syntactic definitions of duality.

\begin{definition}[Syntactic Duality of Session Types]
  \label{def:duality-relation-mu}
  If $\D$ is a relation on $\SType$ then $\Fdual{\D}$ is the relation on $\SType$ defined by:
  \begin{align*}
    \Fdual{\D} & = \{ (\End,\End) \} \\
    & \cup \{ (\inp{T_1}{S_1},\outp{T_2}{S_2}) ~|~ T_1 \sequiv T_2 ~\text{and}~ (S_1,S_2) \in\D \} \\
    & \cup \{ (\outp{T_1}{S_1},\inp{T_2}{S_2}) ~|~ T_1 \sequiv T_2 ~\text{and}~  (S_1,S_2) \in\D \} \\
    & \cup \{ (S_1,\rec{X}{S_2}) ~|~ (S_1,S_2\subst{\rec{X}{S_2}}{X}) \in\D \} \\
    & \cup \{ (\rec{X}{S_1},S_2) ~|~ (S_1\subst{\rec{X}{S_1}}{X},S_2) \in\D ~\text{and}~ S_2\not=\rec{Y}{S_3} \}
  \end{align*}
  A relation $\D$ on $\SType$ is a \emph{session duality} if $\D
  \subseteq \Fdual{\D}$. Duality of session types, $\perp$, is the
  largest session duality. 
\end{definition}

\begin{proposition}[Type Duality Is Tree Duality]
  \label{prop:session-correctness}
  Let $R, S \in\SType$.  Then $R\perp S$ if and only if
  $\treeof R \dualrel \treeof S$.
\end{proposition}
\begin{proof}
  Similar to that of Proposition~\ref{prop:type-equi-tree-equality}.
\end{proof}

This section introduces duality as a relation on session types. It
turns out that, given a session type~$S$, one can construct a session
type $S'$ such that $\treeof{S}\dualrel\treeof{S'}$, or equivalently $S\perp S'$. The next two sections show two
different approaches to the problem, both starting from session types
in syntactic form (Definition~\ref{def:types}).



\section{Duality \`{a} la Bernardi-Hennessy}
\label{sec:duality-bh}

Bernardi and Hennessy~\cite{BernardiH16} observe that the problem of
building a dual session type with naive duality (as explained in Section~\ref{sec:intro}) is caused by free variables in message types. They give a
method for constructing a dual type for an arbitrary session type $S$:
\begin{enumerate}
\item Convert $S$ into an equivalent type $S'$ in which every message
  type is closed. This step is called \emph{message closure} (Definition~\ref{def:mclos}, later).
\item Apply naive duality to $S'$ (Definition~\ref{def:naiveduality}, below).
\end{enumerate}
In this section we present the details of this approach. First we gather the definition of naive duality from Section~\ref{sec:intro}.

\begin{definition}[Naive Duality Function]
  \label{def:naiveduality}
  The naive duality function on session types is inductively defined
  as follows.
  \begin{align*}
    \dual{\inp{T}{S}} & =\, \outp{T}{\dual{S}} &
    \dual{\outp{T}{S}} & =\, \inp{T}{\dual{S}} & 
    \dual{\End} & =\, \End & \dual{X} & =\, X &
    \dual{\rec{X}{S}} & =\, \rec{X}{\dual{S}}
  \end{align*}
\end{definition}

We use the term \emph{tail recursive} for session types in which all
message types are closed. It turns out that these are \emph{not} types
with variables in tail positions only. A counterexample is
$\rec X{\outp{(\inp\Int X)}\End}$ where $X$ occurs in tail position,
but there is a message type that is not closed, namely
$\inp\Int X$.  To define tail recursive types we introduce a type
formation system that essentially keeps track of the free variables in
processes, in such a way that types such as the above are deemed
ill-formed.

\begin{definition}[Tail Recursive Types]
  \label{def:tail-rec-types}
  Let $\varseq$ be a set of type variables.
  The set of \emph{tail recursive types over $\varseq$}, notation
  $\istailrec T$, is defined inductively as follows.
  \begin{gather*}
    \frac{}{\istailrec\Int}
    \qquad
    \frac{}{\istailrec\End}
    \qquad
    \frac{\istailrec[\emptyset]S \quad \istailrec T}{\istailrec{\inp ST}}
    \qquad
    \frac{\istailrec[\emptyset]S \quad \istailrec T}{\istailrec{\outp ST}}
    \qquad
    \frac{X\in\varseq}{\istailrec[\varseq]X}
    \qquad
    \frac{
      \istailrec[\varseq,X]T}{\istailrec{\rec XT}}
  \end{gather*}
  %

  The set of \emph{tail recursive types} is the set of types $T$ such
  that $\istailrec[\emptyset]T$.
\end{definition}

We can easily see that, if $\istailrec{\rec X{\outp ST}}$, then~$X$
does not occur free in $S$. In particular the type
$\rec X{\outp{(\inp\Int X)}\End}$ identified above is not tail recursive.

Gay \& Hole~\cite{GayH05} claim to prove that for all $S\in\SType$, $\dual{S} \perp S$. They use a slightly different
definition of $\perp$ in which types are completely unfolded before
analysing their structure. Unfolding means repeatedly transforming top-level
$\rec XS$ to $S[\rec XS/X]$ until a non-$\mu$ type is exposed.
However, the proof contains the claim that
if the unfolding of $S$ is $\inp{T}{S'}$ then the unfolding of $\dual{S}$ is 
$\outp{T}{\dual{S'}}$, which is not true if $T$
contains type variables. Their proof does, however, show the
following result.


\begin{proposition}[Soundness of Naive Duality for Tail Recursive Types~\cite{GayH05}]
  If $S$ is a tail recursive type, then $S \perp \dual{S}$.
\end{proposition}

This supports the Bernardi-Hennessy approach, because it shows that if a session type can be converted to an equivalent type that is tail recursive, then it is sufficient to apply naive duality to the tail recursive type.

Message closure builds a tail recursive session type
by collecting substitutions $\subst {\rec XS}X$ for each $\rec XS$ type
encountered and applying the accumulated substitution to
messages.
%
%
To define message closure we need the notion of a
sequence of substitutions.

\begin{definition}[Sequence of Substitutions]
  A \emph{sequence of substitutions} is given by the following grammar:
  \begin{equation*}
    \sigma \grmeq \Empty \grmor \subst SX; \sigma
  \end{equation*}
  The \emph{application of a sequence of substitutions} $\sigma$ to a type
  $T$---notation $T\sigma$---is defined as $T\Empty = T$ and
  $T(\subst SX;\sigma) = (T\subst SX)\sigma$.
  A sequence of substitutions $\sigma$ is \emph{closing} for $T$ if $\fv(T\sigma) = \Empty$.
\end{definition}

\begin{definition}[Message Closure \cite{BernardiH16}]
  \label{def:mclos}
  For any type $T$ and sequence of substitutions $\sigma$ closing for $T$, the type
  $\mclo{T}{\sigma}$ is defined inductively by the following rules.
  \begin{align*}
    \mclo{\End}{\sigma} & =\, \End &
     \mclo{X}{\sigma} & =\, X
    \\
    \mclo{\outp{T}{S}}{\sigma} & =\, \outp{(T\sigma)}{\mclo{S}{\sigma}} &
    \mclo{\rec{X}{S}}{\sigma} & =\, \rec{X}{\mclo{S}{\subst{(\rec{X}{S})}{X};\sigma}}
    \\
    \mclo{\inp{T}{S}}{\sigma} & =\, \inp{(T\sigma)}{\mclo{S}{\sigma}} &
  \end{align*}
  Define $\mcl{S}$ as $\mclo{S}{\Empty}$. 
\end{definition}

Bernardi and Hennessy prove that taking the naive dual of the message
closure of a type is sound with respect to a notion of compatibility
based on a labelled transition system for session types. We will prove soundness with respect to regular trees. First, however, we show that if $S\in\SType$ then $\mcl S$ is tail recursive.

The next two lemmas are easily proved by induction.
\begin{lemma}
  \label{lem:free-vars-tailrec}
  If $\istailrec T$, then $\fv(T) \subseteq \varseq$.
\end{lemma}
%

\begin{lemma}[Strengthening]
  \label{lem:strengthening}
  If $\istailrec[\varseq,X] T$ and  $X\notin\fv(T)$, then $\istailrec T$.
\end{lemma}
%
Combining them, we can identify exactly the type variables that occur free in message positions.
\begin{corollary}
  \label{prop:tailrec-restriction}
  If $\istailrec T$, then $\istailrec[\fv(T)]T$.
\end{corollary}
\begin{proof}
  From Lemma~\ref{lem:free-vars-tailrec} we know that
  $\fv(T) \subseteq \varseq$. Use Strengthening
  (Lemma~\ref{lem:strengthening}) repeatedly to remove from~$\varseq$
  type variables not in $\fv(T)$.
\end{proof}

Finally, we reason about $\mclo{T}{\sigma}$.
\begin{lemma}
  \label{lemma:mclosure}
  If $\sigma$ is a closing substitution for $T$, then $\istailrec[\dom(\sigma)]{\mclo T\sigma}$.
\end{lemma}
\begin{proof} 
  Straightforward induction on the definition of $\mclo T\sigma$.
%
%
%
%
%
%
\end{proof}

\begin{corollary}
  If $T$ is closed, then $\mcl T$ is tail recursive.
\end{corollary}
\begin{proof}
  If $T$ is closed, then $\Empty$ is a closing substitution for
  $T$. Lemma~\ref{lemma:mclosure} ensures that
  $\istailrec[\emptyset]{\mclo T\Empty}$, hence
  $\istailrec[\emptyset]{\mcl T}$ by definition.
\end{proof}

\begin{example}
  \label{ex:bh-duality}
  The Bernardi-Hennessy approach to duality applied to our
  running example $S= \rec X{\outp XX}$.
\begin{align*}
  \dual{\mcl{S}} = \dual{\mclo{S}{\Empty}}
  &= \dual{\rec X {\mclo{\outp XX}{\subst SX}}}\\
  &= \dual{\rec X {(\outp{X\subst SX)}{\mclo{X}{\subst SX}}}}\\
  &= \dual{\rec X {\outp{S}{X}}} = \rec X {\dual{\outp SX}} = \rec X {\inp S\dual X} = \rec X {\inp SX}
\end{align*}  
\end{example}

It turns out that the two steps---application of message closure and
the computation of naive duality---can be combined into a single step,
performing message closure during the process of computing the dual
type.
This is captured by the definition below, which constructs the dual of a type in a single pass
over its abstract syntax tree.

\begin{definition}[Duality with On-the-fly Message Closure]
  \label{def:mcduality}
  For any session type $S$ and sequence of substitutions $\sigma$ closing for $S$, the
  session type $\mcduals S$ is defined inductively by the
  following rules.
  \begin{align*}
    \mcduals{\End} & =\, \End &
    \mcduals{X} & =\, X \\
    \mcduals{\outp{T}{S}} & =\, \inp{(T\sigma)}{\mcduals{S}} &
    \mcduals{\rec{X}{S}} & =\, \rec{X}{\mcduals[\subst{(\rec{X}{S})}{X};\sigma]{S}} \\
    \mcduals{\inp{T}{S}} & =\, \outp{(T\sigma)}{\mcduals{S}}
  \end{align*}
  Define $\mcdual S$ as $\mcduals[\Empty] S$.
\end{definition}

\begin{example}
  Here is duality with on-the-fly message closure in action for our
  running example $S= \rec X{\outp XX}$.
  \begin{align*}
    \mcdual S = \mcduals[\Empty]{S}
    &= \rec X {\mcduals[\subst SX]{\outp XX}}\\
    &= \rec X {(\inp{X\subst SX)}{\mcduals[\subst SX]{X}}}\\
    &= \rec X {\inp{S}{\mcduals[\subst SX]{X}}}\\
    &= \rec X {\inp SX}
  \end{align*}
  The economy with respect to the original definition,
  Example~\ref{ex:bh-duality}, should be apparent.
\end{example}

\begin{example}
Consider the problematic type of Bernardi and
Hennessy~\cite{BernardiH14}. Let $S_2 = \rec Y{\outp YX}$ and
$S_1= \rec X{S_2}$. We then have:
\begin{align*}
  \mcdual{S_1} = \mcduals[\Empty]{{S_1}}
  &= \rec X {\mcduals[\subst {S_1}X]{S_2}}\\
  &= \rec X {\rec Y {\mcduals[\subst {S_1}X\subst {S_2}Y]{\outp YX}}}\\
  &= \rec X {\rec Y {\inp {(Y\subst {S_1}X\subst {S_2}Y)}{\mcduals[\subst {S_1}X\subst {S_2}Y] X}}}\\
  &= \rec X {\rec Y {\inp {S_2}{\mcduals[\subst {S_1}X\subst {S_2}Y] X}}}\\
  &= \rec X {\rec Y {\inp {S_2} X}}
\end{align*}
\end{example}
Applying message closure on the fly does not change the tail recursive type that we obtain.
\begin{proposition}
\label{prop:mcdual-dualmcl}
  If $S\in\SType$ then $\mcdual{S}$ is syntactically equal to $\dual{\mcl{S}}$.
\end{proposition}
\begin{proof}
Prove by structural induction on $S$ that for any sequence of substitutions $\sigma$ closing for $S$, $\mcduals[\sigma]{S}$ is syntactically equal to $\mclo{S}{\sigma}$.
\end{proof}

We can now show that duality with on-the-fly
message closure is sound with respect to duality on regular trees.
\begin{proposition}
If $S\in\SType$ then $\treeof{\dual{\mcl{S}}} \dualrel\treeof{S}$.
\end{proposition}
\begin{proof}
Instead of proving this directly, we go via definitions and results from Section~\ref{sec:duality-lm}. Proposition~\ref{prop:mcdual-dualmcl} shows that $\dual{\mcl{S}} = \mcdual{S}$. Proposition~\ref{prop:mcdual-cdual} shows that $\mcdual{S} = \cdual{S}$, where $\cdual{S}$ is defined in Definition~\ref{def:lm-nonewneg}. Therefore $\treeof{\dual{\mcl{S}}} = \treeof{\cdual{S}}$. Finally, Proposition~\ref{thm:lmdual} shows that $\treeof{\cdual{S}}\asymp\treeof{S}$.
\end{proof}



\section{Duality \`{a} la Lindley-Morris}
\label{sec:duality-lm}

The Lindley-Morris definition of the duality
function~\cite{LindleyM16} uses negative type variables $\dual{X}$,
which we therefore add to the syntax in Definition~\ref{def:types}. In
type $\rec{X}{T}$, both the positive variable $X$ and the negative
variable $\dual{X}$ are bound in $T$. Corresponding to this extension,
we generalise the definition of $\treeof{\cdot}$
(Definition~\ref{def:treeof}) so that $\coidual{\cdot}$ (Definition~\ref{def:tree-duality}) is applied to the subtrees that arise from negative variables.

\begin{example}
  Let $S$ be the type $\rec{X}{\outp{X}{\dual X}}$. Let $s = \treeof{S}$ and let $t = \coidual{s}$. Tree $s$ 
  can be depicted as (i) below. To obtain tree $t$, (ii) below, we dualise the root,
  keep $s$ for the left subtree and use the dual of $t$ (that is,
  $s$) for the right subtree (cf.\ the rule
  $\dual{\outp{T}{S}} = \inp{T}{\dual{S}}$ in
  Definition~\ref{def:naiveduality}).
  Substituting $t$ into tree (i) gives tree (iii).
  Tree (iv) shows a few more nodes in the expansion of~$s$.
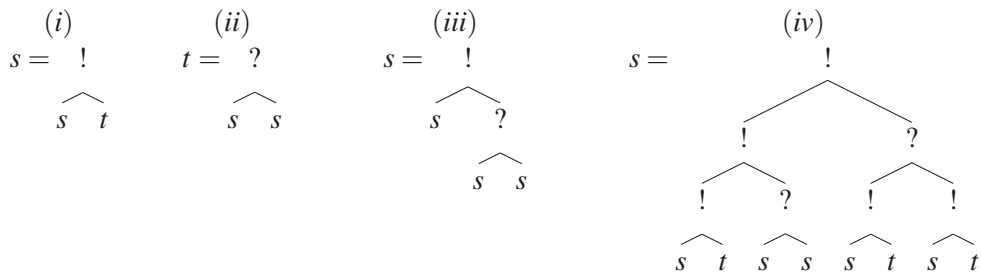
\[
\begin{array}{c@{\extracolsep{10mm}}ccc}
(i) & (ii) & (iii) & (iv) \\
  s=\Tree[.$!$ s t ]
&
  t=\Tree[.$?$ s s ]
&
  s=\Tree[.$!$ s [.$?$ s s ] ]
&
  s=\Tree[.$!$ [.$!$ [.$!$ s t ] [.$?$ s s ] ] [.$?$ [.$!$ s t ] [.$!$ s t ] ] ]
\end{array}
\]
  
\end{example}

The definition of the duality function also requires a particular form
of substitution that exchanges negative variables $\dual X$ and
positive variables~$X$.

\begin{definition}[Negative Variable Substitution]
  The result of substituting $\dual X$ for the free occurrences of $X$
  in $T$---notation $T\lmsubst{\dual X}{X}$---is defined inductively as
  follows.
  \begin{align*}
    X\lmsubst{{\dual X}}{X} & =\, {\dual X} &
    \Int\lmsubst{{\dual X}}{X} & =\, \Int \\
    \dual X\lmsubst{{\dual X}}{X} & =\, X &
    \End\lmsubst{{\dual X}}{X} & =\, \End \\
    Y\lmsubst{{\dual X}}{X} & =\, Y  \text{~~if } Y\neq X &
    (\inp{S}{T})\lmsubst{{\dual X}}{X} & =\, \inp{(S\lmsubst{{\dual X}}{X})}{T\lmsubst{{\dual X}}{X}} \\
    \dual{Y}\lmsubst{{\dual X}}{X} & =\, \dual{Y} &
    (\outp{S}{T})\lmsubst{{\dual X}}{X} & =\, \outp{(S\lmsubst{{\dual X}}{X})}{T\lmsubst{{\dual X}}{X}} \\
    & & (\rec{Y}{S})\lmsubst{{\dual X}}{X} & =\, \rec{Y}{S\lmsubst{{\dual X}}{X}}
  \end{align*}
\end{definition}

\begin{definition}[Lindley-Morris Duality, Original Version \cite{LindleyM16}]
  \label{def:lmduality}
  \begin{align*}
    \lmdual(\End) & =\, \End &
    \lmdual(X) & =\, \dual{X} \\
    \lmdual(\inp{T}{S}) & =\, \outp{T}{\lmdual(S)} &
    \lmdual(\dual{X}) & =\, X \\
    \lmdual(\outp{T}{S}) & =\, \inp{T}{\lmdual(S)} &
    \lmdual(\rec{X}{S}) & =\, \rec{X}{(\lmdual(S)\lmsubst{\dual{X}}{X})}
  \end{align*}
\end{definition}
\begin{example}
%
\begin{align*}
  \lmdual{(\rec X{\outp XX})}
  &= \rec X {\lmdual{((\outp XX)}\subst{\dual X}X)} = \rec X {(\inp X \lmdual(X))\subst{\dual X}X}\\
  &= \rec X {(\inp X \dual X)\subst{\dual X}X} = \rec X {\inp {X\subst{\dual X}X} {\dual X\subst{\dual X}X}}\\
  &= \rec X {\inp {\dual X} {\dual X\subst{\dual X}X}} = \rec X {\inp {\dual X} X}
\end{align*}
\end{example}

This definition of duality is sound with respect to trees.
\begin{proposition}
\label{prop:lmdual-sound}
If $S\in\SType$ then $\lmdual(S)\asymp S$.
\end{proposition}
\begin{proof}
This is one of the results that we have mechanized in Agda (Section~\ref{sec:mechanized-results}).
\end{proof}

There is an alternative formulation of
Lindley-Morris duality that works with conventional substitution (Definition~\ref{def:substitution} with the additional clause $\dual{Y}\subst{S}{Z} = \dual{Y}$, i.e., no substitution for negative variables). The idea is that the
(bound) occurrences of $X$ in the dual of $\rec XS$ are occurrences
not of $X$ (which stands for $S$) but of $\dual X$ (which stands for
the dual of $S$). So we first substitute $\dual X$ for $X$ in $S$ and
only then apply the duality function.

\begin{definition}[Lindley-Morris Duality, Polished]
  \label{def:lmduality-polished}
  \begin{align*}
    \lmdualp(\End) & =\, \End &
    \lmdualp(X) & =\, \dual{X}\\
    \lmdualp(\inp{T}{S}) & =\, \outp{T}{\lmdualp(S)} &
    \lmdualp(\dual X) & =\, X \\
    \lmdualp(\outp{T}{S}) & =\, \inp{T}{\lmdualp(S)} &
    \lmdualp(\rec{X}{S}) & =\, \rec{X}{\lmdualp(S\subst{\dual{X}}{X})}
  \end{align*}
\end{definition}
\begin{example}
\begin{align*}
  \lmdualp{(\rec X{\outp XX})}
  &= \rec X {\lmdualp{((\outp XX)\subst{\dual X}X)}}\\
  &= \rec X {\lmdualp{(\outp{\dual X}{\dual X})}}
   = \rec X {\inp{\dual X}{\lmdualp(\dual X)}}
   = \rec X {\inp{\dual X}X}
\end{align*}
\end{example}

\begin{proposition}
  For any session type $S$, $\lmdual(S)$ is syntactically equal to
  $\lmdualp(S)$.
\end{proposition}
\begin{proof}
By structural induction on $S$, using a lemma that $\lmdualp$ commutes with substitution.
\end{proof}


If we are constructing the dual of a session type that contains no negative variables, we might want to avoid introducing negative variables when dualising a recursive type $\rec{X}{S}$. We can achieve this by using Definition~\ref{def:lmduality-polished} and, at the end, replacing all
occurrences of $\dual X$ (there are no bound occurrences of $\dual X$)
by the original type $\rec XS$. 

\begin{definition}[Lindley-Morris Duality, Yielding No New Negative Variables]
  \label{def:lm-nonewneg}
  \begin{align*}
    \cdual{\End} & =\, \End &
    \cdual{X} & =\, \dual{X}
    \\
    \cdual{(\inp{T}{S})} & =\, \outp{T}{\cdual{S}} &
    \cdual{\dual X} & =\, X
    \\
    \cdual{(\outp{T}{S})} & =\;\inp{T}{\cdual{S}} &
    \cdual{(\rec{X}{S})} & =\;
      \rec{X}{((\cdual{S\subst{\dual{X}}{X})}\subst{\rec XS}{\dual X})}
  \end{align*}
\end{definition}
\begin{example}
\begin{align*}
  \cdual S = \cdual{(\rec X{\outp XX})}
  &= \rec X {\cdual{((\outp XX)\subst{\dual X}X)}\subst{S}{\dual X}} \\
  &= \rec X {\cdual{(\outp{\dual X}{\dual X})}\subst{S}{\dual X}} \\
  &= \rec X {(\inp{\dual X}{\cdual{\dual X}})\subst{S}{\dual X}} \\
  &= \rec X {(\inp{\dual X}X) \subst{S}{\dual X}} = \rec X {\inp{S}X}
\end{align*}
\end{example}

This version of the Lindley-Morris definition coincides with the Bernardi-Hennessy definition.

\begin{proposition}
\label{prop:mcdual-cdual}
For any session type $S$, $\mcdual{S}$ is
syntactically equal to $\cdual{S}$.
\end{proposition}
\begin{proof}
If $\sigma$ is a sequence of substitutions $\subst{T_1}{X_1}\ldots\subst{T_n}{X_n}$ then let $\dual{\sigma} = \subst{T_1}{\dual{X_1}}\ldots\subst{T_n}{\dual{X_n}}$ and $\hat{\sigma} = \subst{\dual{X_1}}{X_1}\ldots\subst{\dual{X_n}}{X_n}$. Prove by structural induction on $S$ that for any sequence of substitutions $\sigma$ closing for $S$, $\mcdual{S,\sigma} = \cdual{(S\hat{\sigma})}\dual{\sigma}$. The result follows by taking $\sigma = \epsilon$.
\end{proof}

Finally, the Lindley-Morris definition is sound with respect to regular trees.

\begin{proposition}
  \label{thm:lmdual}
  If $S\in\SType$ then $\treeof{\cdual{S}}\asymp\treeof{S}$.
\end{proposition}
\begin{proof}
First show that $\mathcal{D} = \{ (S,\cdual{S}) \mid S\in\SType \}$ is a session duality (Definition~\ref{def:duality-relation-mu}). This establishes $\cdual{S} \perp S$. Then use Proposition~\ref{prop:session-correctness}.
\end{proof}

The substitutions in Definition~\ref{def:lm-nonewneg}, or equivalently in the definition of message closure (Definition~\ref{def:mcduality}) increase the size of the type. A simple example shows that this increase can be at least quadratic. If $S = \rec{X}{\inp{X}{\cdots\inp{X}{X}}}$ with $n$ inputs, so that the size of $S$ is $n+2$, then $\mcl{S} = \rec{X}{\inp{S}{\cdots\inp{S}{X}}}$ of size $n(n+2)+2$. In contrast, Definitions~\ref{def:lmduality} and~\ref{def:lmduality-polished} preserve the size of the type because they only substitute variables for variables. In an implementation of a programming language with session types, it is possible to avoid computational issues resulting from these syntactic size increases, by working with a graph representation of regular trees.

\section{Mechanized Results}
\label{sec:mechanized-results}

We mechanized some of the results of the paper in Agda and are working
towards a full mechanized account of all results. For accessibility,
we paraphrase the definitions in standard mathematical notation rather
than Agda syntax. Cognoscenti may explore the Agda source code
corresponding to the development in this section in file
\texttt{Duality.agda} at 
\url{https://github.com/peterthiemann/dual-session}.

The baseline for the mechanization is the coinductive formalization of
session types (Definition~\ref{def:co-sessions}), which we consider as
the ground truth.
In this setting, a session type is a
 potentially infinite tree as contained in the greatest fixpoint $\SType^\infty$ of function $\Sgen$.
 \begin{definition}[Coinductive Session Types]
\label{def:co-sessions}
   \begin{align*}
     \Sgen(\ST) &= \{ \End \}
               \cup \{ \outp{T}S, \inp{T}S \mid S \in \ST, T \in \{\Int\} \cup \ST \}
   \end{align*}
 \end{definition}
 Defining duality for coinductive session types is a straightforward
 corecursively defined function which we call $\coidual\cdot$, reusing the name from Definition~\ref{def:tree-duality} because it implements that function.
 \begin{definition}[Corecursive Duality Function]
   \begin{align*}
     \coidual \End & = \End &
     \coidual {\outp T S} &= \inp T {\coidual S} &
     \coidual {\inp T S} &= \outp T {\coidual S} 
   \end{align*}
 \end{definition}
It is also straightforward to define the duality relation (cf.\
Def.~\ref{def:duality-relation-mu}) as the greatest fixpoint $(\perp)$ of $\Fdual\cdot$.
\begin{definition}[Duality on Coinductive Session Types]
  If $\D$ is a binary relation on tree types, then
  \begin{align*}
    \Fdual\D &= \{ (\End, \End) \} \cup \{ (\outp T S, \inp T
               {S^\perp}) , (\inp T S, \outp T {S^\perp}) \mid (S, S^\perp) \in \D \}
  \end{align*}
\end{definition}
Given these definitions, it is easy to show that the corecursive
duality function is sound and complete with respect to the duality
relation (cf.\ Proposition~\ref{prop:session-correctness}).
\begin{proposition}
  $S \perp S'$ if and only if $S' = \coidual S$. 
\end{proposition}

To formalize session types inductively, we insist that $\mu$-types are
in normal form where there are no consecutive $\mu$-abstractions, i.e.,
no subterms of the form $\rec X{\rec Y S}$, and the body of a $\mu$ is
never a variable. Normal forms are contractive by construction and
every contractive session type (according to Definition~\ref{def:types}) can
be converted to its equivalent normal form by repeatedly coalescing
subterms of the form $\rec X{\rec Y S}$
to $\rec X {S\subst X Y}$ and transforming subterms of the form $\rec
X Y$ to $Y$, assuming $X\ne Y$. The Agda formalization enforces normal forms using
two mutually recursive syntactic categories, $S$ and $S'$, for session types:
\begin{align*}
  S &\grmeq S' \grmor \rec X S' \grmor X \grmor \dual{X} &
  S' &\grmeq \End \grmor \outp T S \grmor \inp T S &
  T &\grmeq \Int \grmor S
\end{align*}
For this representation, we state various definitions of duality as
shown in Sections~\ref{sec:duality-bh} and~\ref{sec:duality-lm}.
Next, we define an embedding $\embed\cdot$ from $\SType$ to tree types by
unfolding the recursion. This function corresponds to the $\treeof\cdot$
function (Definition~\ref{def:treeof}).
\begin{align*}
  \embed {\rec X S'} & = \embed{ S'\subst{\rec X S'} X}' &
  \embed{\End}' &= \End &
  \embed{\outp T S}' &= \outp{\embed T}{\embed S} &
  \embed{\inp T S}' &= \inp{\embed T}{\embed S} &
  \embed{\Int} &= \Int
\end{align*}
This definition is mutually recursive ($\embed{\cdot}$ applies to $S$
and $\embed\cdot'$ applies to $S'$) and it is guarded (i.e., it
yields a proper, potentially infinite term) because $\embed{\cdot}'$
always yields a top-level constructor.

We successfully mechanised a range of results from this paper among
them Proposition~\ref{prop:lmdual-sound}, restated here with the embedding
function. 

\begin{proposition}
\label{thm:mech-lmdual}
  For all $S \in \SType$,
  $\coidual{\embed S} = \embed{\lmdual(S)}$.
\end{proposition}



\section{Conclusion}
\label{sec:conclusion}

We surveyed the competing definitions of session type duality in the
presence of recursion. Starting from an interpretation of session types as trees, and a duality relation on trees, we establish soundness of the Bernardi-Hennessy and the Lindley-Morris definitions of duality on syntactic session types.  
We further come up with streamlined versions of these
definitions and justify the original flawed definition of duality (naive duality) when restricted
to tail recursive session types. We have mechanized some results in Agda, and are working on mechanizing the others.

In summary, we have tied up the remaining loose ends in the definition of duality of session types. Many of the issues in prior work are caused by syntax, namely by reliance on $\mu$-types to express recursion. Taking a standard interpretation of recursive types as regular trees, and the corresponding formalization by coinductive definitions in Agda, is effective   in proving the soundness of syntactic definitions.

\noindent\textbf{Acknowledgements.} Simon Gay was partially supported
by the UK EPSRC grant EP/K034413/1 ``From Data Types to Session Types:
A Basis for Concurrency and Distribution'' and by the EU Horizon 2020
MSCA-RISE project 778233 ``BehAPI: Behavioural Application Program
Interfaces''. Vasco T.\ Vasconcelos was supported by FCT through the LASIGE Research
Unit, ref.\ UIDB/00408/2020, and by COST Action CA15123 EUTypes. We thank Sam Lindley and Garrett Morris for discussions.
%


\bibliographystyle{eptcs}
\bibliography{main}

\end{document}